\documentclass[12pt,reqno]{amsart}
\usepackage{amsmath,amssymb,enumerate,graphicx}

\numberwithin{equation}{section} \numberwithin{figure}{section}
\numberwithin{table}{section}

\newtheorem{theorem}{Theorem}[section]
\newtheorem{proposition}{Proposition}[section]
\newtheorem{corollary}{Corollary}[section]

\newtheorem{assumption}{Assumption}[section]
\newtheorem{lemma}{Lemma}[section]
\newtheorem{example}{Example}[section]

\begin{document}

\vspace*{1mm}

\begin{center}
{\Large Log-supermodularity of weight functions and the loading
monotonicity of weighted insurance premiums}

\bigskip

{\large Hristo S.~Sendov$^{a}$,
Ying Wang$^{b}$,
and Ri\v cardas Zitikis$^{a,*}$}

\medskip

$^{a}$\textit{Department of Statistical and Actuarial Sciences,
University of Western Ontario, London, Ontario N6A 5B7, Canada}

$^{b}$\textit{Department of Applied Mathematics,
University of Western Ontario, London, \break Ontario N6A 5B7, Canada}

\medskip
(E-mails:
\texttt{hssendov@stats.uwo.ca},~~~\texttt{ywang767@uwo.ca},~~~\texttt{zitikis@stats.uwo.ca})

\end{center}

\medskip

\noindent
\rule{\linewidth}{0.2mm}

\noindent
\textbf{Abstract}
\medskip

The paper is motivated by a problem concerning
the monotonicity of insurance premiums with respect to
their loading parameter: the larger the parameter, the larger the insurance
premium is expected to be. This property, usually called
loading monotonicity, is satisfied by premiums that appear in
the literature. The increased interest in constructing new insurance premiums
has raised a question as to what weight functions would produce
loading-monotonic premiums. In this paper we demonstrate
a decisive role of log-supermodularity
in answering this question. As a consequence, we establish -- at a stroke --
the loading monotonicity of a number of well-known insurance premiums
and offer a host of further weight functions,
and consequently of premiums,
thus illustrating the power of the herein suggested methodology
for constructing loading-monotonic insurance premiums.

\bigskip

\noindent
\textit{JEL Classification:}
\begin{quote}
C02 - Mathematical Methods\\
C44 - Statistical Decision Theory; Operations Research \\
C51 - Model Construction and Estimation\\
D81 - Criteria for Decision-Making under Risk and Uncertainty
\end{quote}

\medskip


\noindent
\rule{\linewidth}{0.2mm}

\vfill

\noindent
$^{*}${\small Corresponding author:
tel.: +1 519 432 7370;
fax.: +1 519 661 3813;
e-mail: \texttt{zitikis@stats.uwo.ca}}

\newpage

\section{Introduction}
\label{section-1}

Let $(\Omega , \mathcal{A}, \mathbf{P})$ be a probability triplet. The integral
$\mathbf{E}[X]=\int_{\Omega}X(\omega )\mathbf{P}(d\omega )$
is known in the statistical literature as the expectation (or mean)
of the `random variable' $X:\Omega \to \mathbf{R}$, and also
known  in the actuarial literature as the net premium of the `loss variable' $X$.
In the latter case, $X$ usually takes on non-negative values.
In order to avoid mathematical trivialities of constant worry,
throughout this paper except when explicitly noted otherwise,
we work with positive random variables (i.e., $X>0$ almost surely)
and denote their class by $\mathcal{X}_{+} $.

The net premium $\mathbf{E}[X]$ is the very minimum that
the insurer needs to charge customers for agreeing to accept their
risks and have a chance to remain solvent. In fact, in order to meet
various financial obligations, the insurer charges
a larger amount; denote it by $H[X]$. This defines a functional
$H:\mathcal{X}_{+}\to [0,\infty ] $, called premium calculation principle
or, simply, premium. Any premium $H$ such that
$H[X]\ge \mathbf{E}[X] $ for all $ X\in \mathcal{X}_{+}$
is called non-negatively loaded or, simply, loaded.
Loaded premiums are constructed using various actuarial considerations
and mathematical techniques, and we shall next briefly discuss two of them.

First and arguably one of the oldest
general techniques for constructing loaded premiums is
based on the fact that the net premium $\mathbf{E}[X]$ can be
written -- using the Fubini theorem or the integration by parts formula --
as the integral $\int_0^{\infty}\mathbf{P}[X>x]dx$. (Recall that
$X\in \mathcal{X}_{+}$; for real-valued $X$, we would have
a Choquet integral; see Denneberg (1994).)  Choose
any function $g:[0,1] \to [0,\infty )$ such that $g(s)\ge s$ for
all $0\le s \le 1$ and modify the above integral into
$\int_0^{\infty}g(\mathbf{P}[X>x])dx$, which is a loaded premium,
known  in the actuarial literature as the distortion, or Wang's, premium
(e.g., Denuit \textit{et al.}, 2005).
The function $g:[0,1] \to [0,\infty )$ is called
the distortion function and usually depends on a parameter,
called the distortion parameter, which governs the amount of loading
contained in the distortion premium.

The second and fairly recent avenue for constructing loaded premiums
has been suggested by Furman and Zitikis (2008).
Just like in the case of the distortion premium, the construction
starts with the net premium $\mathbf{E}[X]$ but this time transforming
the integrator $\mathbf{P}(d\omega )$
with a `weight' function $w:[0,\infty ) \to [0,\infty )$.
This approach gives a new probability measure
$\mathbf{P}_w(d\omega )=w(X(\omega ))
 \mathbf{P}(d\omega )/ \mathbf{E}[w(X)]$,
assuming of course that the `normalizing' expectation
$\mathbf{E}[w(X)]$ is non-zero, that is, positive.
The weighted premium is (Furman and Zitikis, 2008)
\[
H_w[X]=\int_{\Omega}X(\omega )\mathbf{P}_w(d\omega ).
\]

In many special cases of $H_w[X]$ that we find in the literature
(see, e.g., Furman and Zitikis, 2008, 2009), the weight function $w$ is
indexed with a parameter, which we denote by $\lambda $. The parameter controls
the amount of loading and is therefore called the loading parameter.
After a re-parametrization if necessary, the loading parameter $\lambda $
can be assumed to be in $(0,\infty )$, with the limiting value
of $H_w[X]$ when $\lambda \downarrow 0$
associated with the net premium $\mathbf{E}[X]$.

From now on, therefore, instead of $w(x)$ we work with weight functions
$x\mapsto w(\lambda, x)$ indexed by $\lambda \in (0,\infty ) $.
It therefore becomes natural to start using the notation
$H[\lambda,X]$ instead of the more cumbersome $H_{w(\lambda, \cdot)}[X]$.
That is, let
\[
H[\lambda,X]
={ \mathbf{E}[Xw(\lambda, X)] \over \mathbf{E}[w(\lambda, X)] }
\]
and denote
$\Lambda[w,X]=\{\lambda \in (0,\infty ) :\, \mathbf{E}[Xw(\lambda, X)]<\infty
\textrm{ and } \mathbf{E}[w(\lambda, X)]>0 \}$, which is usually
an interval, finite or infinite, depending on the cumulative distribution
function (cdf) $F$ of the loss variable $X$.
From the intuitive point of view, we expect that
the larger the parameter $\lambda $ is, the larger the value of $H[\lambda,X]$ is.
This monotonicity property may not always hold as it depends on the weight function
$x\mapsto w(\lambda, x)$. Later in this paper we shall specify
conditions under which the loading monotonicity of $H[\lambda,X]$ holds.

A special case of the loading monotonicity property, and
one of the basic requirements for $H[\lambda,X]$ to satisfy,
is the aforementioned non-negative loading property:
\[
H[\lambda,X]\ge \mathbf{E}[X] \quad \textrm{for all} \quad
(\lambda, X)\in (0,\infty )\times \mathcal{X}_{+}.
\]
The property is satisfied whenever the weight function
$x\mapsto w(\lambda, x)$ is non-decreasing, which can be infered
from a classical result of Lehmann (1966) stating that the bound
$\mathbf{E}[ u(X) v(X)] \ge \mathbf{E}[ u(X)] \mathbf{E}[v(X)]$
holds for all non-decreasing (Borel) functions $u$ and $v$
for which the expectations are well-defined and finite.
In fact, this property is well-known in Mathematical Analysis
under the name of
Chebyshev's integral inequality (see, e.g.,
Pe\v{c}ari\'{c} \textit{et al.}, 1992;
Mitrinovi\'{c} \textit{et al.}, 1993; and references therein)
and has been widely utilized when solving a number of problems
in Economics and Finance (see, e.g., Broll \textit{et al.}, 2010;
Egozcue \textit{et al.}, 2010; and references therein) where
non-negativity of the covariance $\mathbf{Cov}[ u(X), v(X)]$
plays a pivotal role.

The following assumption now becomes natural, and we let it hold
throughout the paper without explicitly mentioning it.

\begin{assumption}\rm \label{assumption-onw}
For every (loading) parameter $\lambda > 0$,
the weight function $x\mapsto w(\lambda, x)$
is non-decreasing and non-negative,
and we also assume that the function is
Borel-measurable so that expectations are calculable.
\end{assumption}

The rest of the paper is organized as follows.
In Section \ref{section-2} we establish a fundamental for this paper
result stating that log-supermodularity of the function
$(\lambda, x)\mapsto w(\lambda, x)$ implies loading monotonicity
of the premium $H[\lambda,X]$.
In Section \ref{section-3} we provide a set of parametric weight
functions that either lead to known insurance premiums or to new ones,
and we also verify log-supermodularity of the functions thus establishing
loading monotonicity of the corresponding weighted premiums.
In Section \ref{section-4} we work out conditions under which loading monotonicity
turns into loading strict-monotonicity. In Section \ref{section-5} we
specify conditions under which the function $\lambda \mapsto H[\lambda,X]$
is right-continuous, and even continuous. Some results of technical nature
are relegated to Section \ref{section-6}.

Given the above outline, one may wonder if
our research of monotonicity and continuity of
the function $\lambda \mapsto H[\lambda,X]$ has been driven by
mathematical curiosity or actuarial considerations.
The answer is `both'. Originally, our interest was inspired by
an insurance-related problem, which subsequently
brought in a number of interesting mathematical issues.

To explain the original problem, assume that an insurer favors using
$H[\lambda,X ]$ in all in-house premium calculations and thus wishes to
convert into it all the other premiums $\pi[X]$ in use.
Hence, for each such $\pi[X]$, the insurer wishes to know $\lambda $
such that $H[\lambda,X]=\pi[X]$. Depending on the form of $H[\lambda,X ]$,
there might be several values of $\lambda $
that give the equality $H[\lambda,X ]=\pi[X]$, or
there might be none such $\lambda $. If at least one $\lambda $ exists,
then one may still wish to know whether this $\lambda $ is the only one or not.
Answering such questions, naturally, relies on monotonicity and
continuity properties of the function $\lambda \mapsto H[\lambda,X]$,
and they in turn depend on those of the function $(\lambda, x)\mapsto w(\lambda, x)$
coupled with properties of the loss variable $X$.
Sorting out these issues in detail is our goal in the present paper.
In the context of the aforementioned distortion premium,
the problem has been posed and justified from the actuarial point of view
by Jones and Zitikis (2007).

\section{Log-supermodularity and loading monotonicity}
\label{section-2}

To formulate the main result of this section, which is
Theorem \ref{th-1} below, we need to recall some definitions.
Function $(\lambda, x)\mapsto w(\lambda, x)$ is called log-supermodular if
$(\lambda, x)\mapsto \log( w(\lambda, x))$ is supermodular,
which is equivalent to saying that
$L(\lambda, x) = - \log( w(\lambda, x))$ is submodular.
The function $(\lambda, x)\mapsto L(\lambda, x) $ is submodular if
\begin{align}
\label{boza-1}
L(\theta, x_1) + L(\lambda, x_2) \le L(\theta, x_2) + L(\lambda, x_1)
\end{align}
whenever $\theta \le \lambda $ and $x_1 \le x_2$.

The way we have here presented the definition of submodularity
is to facilitate computations. The standard
way is in the form of the bound
$L(\boldsymbol{y} \wedge \boldsymbol{z})
+ L(\boldsymbol{y} \vee \boldsymbol{z})
\le L(\boldsymbol{y}) + L(\boldsymbol{z})$,
where the minimum $\boldsymbol{y} \wedge \boldsymbol{z}$ and the maximum
$\boldsymbol{y} \vee \boldsymbol{z}$ between the vectors
$\boldsymbol{y}=(y_1,y_2) $ and $\boldsymbol{z}=(z_1,z_2) $
are taken coordinatewise. For details on submodular, supermodular, and other
related functions, we refer to, for example,
Fujishige (1991), Narayanan (1997), and Topkis (2001).
A number of elementary ways for constructing submodular functions are
listed in Table I on p.~312 of Topkis (1978).
Note also that if the function $L$ is sufficiently smooth,
then its submodularity is equivalent to the bound
$(\partial^2/ \partial \lambda \partial x) L(\lambda, x) \le 0$
for all $(\lambda, x) $. We shall utilize this criterion frequently in
this paper.

The following theorem is a fundamental result in the context of the present paper.

\begin{theorem} \label{th-1}
If $(\lambda, x)\mapsto w(\lambda, x)$ is log-supermodular,
then $\lambda \mapsto H[\lambda,X]$ is non-decreasing.
\end{theorem}

To prepare for the proof of Theorem \ref{th-1}, we need to recall
weighted distributions and some of their properties
(see, e.g., Furman and Zitikis, 2008, and references
therein). Thus, assume that we are dealing with a non-decreasing function $w$,
indexed or not. Denote the cdf of $X \in \mathcal{X}_+$ by $F$.
The weighted cdf $F_w$ is defined by
\[
F_w(x) ={\mathbf{E}[ \mathbf{1}\{X\le x\} w(X)] \over \mathbf{E}[w(X)]},
\]
where $\mathbf{1}\{S\}$ is equal to $1$ if statement $S$ is true and $0$
otherwise. Note the equation
$F_w(x)=\int_{\Omega} \mathbf{1}\{X(\omega )\le x\}\mathbf{P}_w(d\omega )$
connects $F_w$ and $\mathbf{P}_w$ in the same way
as the original $F$ and $\mathbf{P}$ are connected by the usual definition
of the cdf $F(x)=\mathbf{P}[X\le x]$.
The support of $F_w$ is $[0,\infty)$.

The importance of the weighted cdf in the current context
is reflected by the fact that if $X_w$ is a random variable with
the cdf $F_w$, then the weighted premium $H_w[X]$ is
the mean $\mathbf{E}[X_w] $. Indeed, using the Fubini theorem, we have that
\begin{align}
H_w[X]
= \int_{[0,\infty)} \frac{\mathbf{E}[\mathbf{1}\{X > x \}w(X)]}{\mathbf{E}[w(X)]} dx
= \int_{[0,\infty)} (1-F_w(x)) dx
=\mathbf{E}[X_w].
\label{eqn-EH}
\end{align}
Equation (\ref{eqn-EH}) plays a major role in establishing Theorem \ref{th-1}
as well as some other results in later sections.
The following properties are known (see, e.g., Furman and Zitikis, 2008,
and references therein) and will be used a number of times
in this paper:
\begin{itemize}
\item
For any two non-decreasing and non-negative functions $u$ and $w$, we have
\begin{equation}
F_{uw} = (F_u)_w=(F_w)_u,
\label{w-cdf-1}
\end{equation}
where $F_{uw}$ is the weighted cdf
corresponding to the product $u(x)w(x)$ of the two functions $u(x)$ and $w(x)$.
\item
For any non-decreasing and non-negative function $w$, we have
$F_{w} \le F$ and, consequently,
for any non-decreasing and non-negative functions $w$ and $ h$,
we have
\begin{equation}
F_{hw} \le F_w \le F.
\label{w-cdf-2}
\end{equation}
\end{itemize}

\begin{proof}[\bf Proof of Theorem \ref{th-1}]

We start by showing that
$L(\lambda, x) = - \log( w(\lambda, x))$ is
submodular if and only if for every pair $\theta \le \lambda$
 there is a non-decreasing function
$h= h_{\theta, \lambda}$ such that
\begin{equation}
w(\lambda,x) = h(x)w(\theta,x).
\label{condition-1}
\end{equation}
(We do not know if this reformulation of log-supermodularity
has been noted in the literature, but we find it invaluable
in the context of the present paper.)
Assuming for the time being that decomposition (\ref{condition-1}) holds,
we next show how to utilize it for proving Theorem \ref{th-1}.

Fix any pair $\theta \le \lambda$ and let $h$
be a non-decreasing function whose existence is postulated
by decomposition (\ref{condition-1}).
In view of bounds (\ref{w-cdf-2}), we have $F_{w(\lambda, \cdot)} \le
F_{w(\theta,\cdot)}$. Using equations (\ref{eqn-EH}),
we therefore have
$\mathbf{E}[X_{w(\lambda, \cdot)}] \ge  \mathbf{E}[X_{w(\theta,\cdot)}]$
and thus $H[\lambda,X] \ge H[\theta,X]$, which is
the claim of Theorem \ref{th-1}.
We are left to demonstrate the equivalence of submodularity of
$L(\lambda, x)$ and the existence of a non-decreasing function $h$ such that
decomposition (\ref{condition-1}) holds.

Fix any pair $\theta \le \lambda$.
The submodularity of $L(\lambda, x)$ implies that, for $x_1\le x_2$,
\begin{equation}
\label{submod-ineq}
w(\theta, x_1)w(\lambda, x_2) \ge w(\theta, x_2)w(\lambda, x_1).
\end{equation}
Define $y_{\Delta}= \sup \{x \in [0, \infty) : w(\Delta , x) = 0\}$
for $\Delta \in \{\theta, \lambda\}$.
If $y_\theta = -\infty$, that is, if
$w(\theta, x) > 0$ for all $x$, then the existence of the
aforementioned function $h(x)$ is trivial. Otherwise, for every $x_1$ with
$w(\theta, x_1)=0$, we choose $x_2 > y_\theta$, and then
bound (\ref{submod-ineq}) implies $w(\lambda, x_1)=0$ thus
showing that $y_\theta \le y_\lambda$. We can now define
the sought after function $h$ as follows:
\[
h(x) = \left\{
\begin{array}{ll}
0 & \mbox{ when } x > 0 \mbox{ is such that } w(\theta, x) = 0, \\
w(\lambda, x)/w(\theta, x) & \mbox{ otherwise. }
\end{array}
\right.
\]
Note that $h$ is non-decreasing, as is implied by
bound (\ref{submod-ineq}).

Conversely, suppose that there is a non-decreasing function $h
= h_{\theta, \lambda}$ such that $w(\lambda, x) = h(x)
w(\theta, x)$. Then $y_\theta \le y_\lambda$ and
$w(\theta, y_\theta) = 0$ imply  $w(\lambda, y_\theta) = 0$. For
any $x_1 \le x_2$ we consider several cases depending
on the position of $x_1$ and $ x_2$ relative to $y_\theta$ and $y_\lambda$,
and easily conclude that bound (\ref{submod-ineq}) holds. This shows that
$L(\lambda, x) = - \log( w_\lambda (x))$ is submodular,
and thus completes the proof of Theorem \ref{th-1}.
\end{proof}

\section{Seven classes of weight functions}
\label{section-3}

To show the encompassing nature and power of Theorem \ref{th-1},
we next present seven illustrative examples of log-supermodular functions.
The first three examples correspond to the Esscher,
conditional tail expectation (CTE), and the Kamps premiums,
which are exceptionally well-known in the actuarial literature.
The other four examples are `mathematical inventions', but their
varying convex or concave shapes hint at potential usefulness. Moreover,
we have to note that the mathematical literature is not particularly generous with
examples of log-supermodular functions, and thus the ones
that we offer here can be viewed as contributions to the area of
Mathematical Analysis and especially of Function Theory.
Since the four weight functions are somewhat complex,
we shall supplement their definitions with graphs.

\begin{example}\rm\label{example1.1}\hspace*{1mm}
\begin{enumerate}
  \item Let $w_1(\lambda, x)=e^{\lambda x}$.
  The corresponding
  weighted premium $H[\lambda,X]$ is known in the literature
  as the Esscher premium  (see, e.g., Denuit \textit{et al.}, 2005, and
  references therein).
  \item  Let $w_2(\lambda,x)=\mathbf{1}\{x> \lambda \}$.
  The corresponding
  weighted premium $H[\lambda,X]$, which can be written as
  $\mathbf{E}[X|X> \lambda]$, is known as
  the conditional tail expectation (CTE).
  We refer to Denuit \textit{et al.} (2005)
  for detailed information on the CTE premium.
  \item  Let
  $w_3(\lambda,x) = 1-e^{-x/\lambda }$.
The corresponding
  weighted premium $H[\lambda,X]$ is known
  as the Kamps premium (Kamps, 1998; see also Furman and Zitikis, 2008, 2009).
  \item  Let
  $w_4(\lambda,x) = e^{((1+x)^\lambda -1) / \lambda}-x $.
  The fact that $x\mapsto w_4(\lambda, x)$ is non-decreasing
  for every $\lambda > 0$ is easy to establish.
\begin{figure}[h!]
\includegraphics[width=10cm]{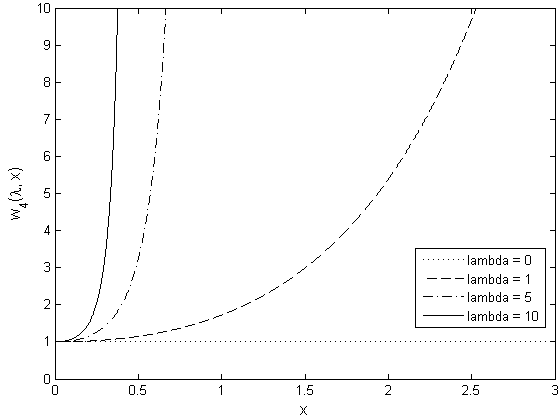}
\caption{The function $x\mapsto w_4(\lambda, x)$. When $x\downarrow 0$,
then $w_4(\lambda, x)$ converges to $1$ for every $\lambda >0$.}
\label{fig-w4}
\end{figure}
    \item  Let $w_5(\lambda,x) = \big((1+\lambda)^x -1 \big) / (x\lambda) $.
  It is easy to show that $x\mapsto w_5(\lambda, x)$ is non-decreasing
  for every $\lambda > 0$.
\begin{figure}[h!]
\includegraphics[width=10cm]{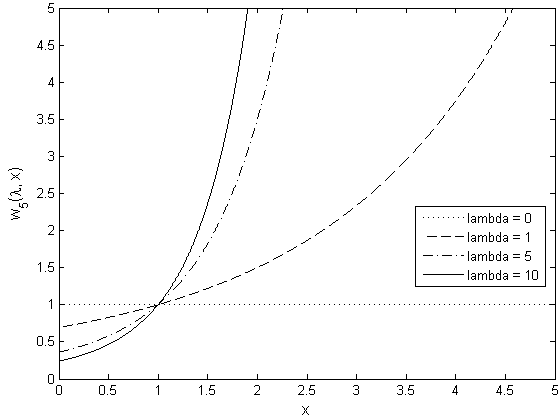}
\caption{The function $x\mapsto w_5(\lambda, x)$. When $x\downarrow 0$,
then $w_5(\lambda, x)$ converges to $\log(1+\lambda)/\lambda $ for every $\lambda >0$.}
\label{fig-w5}
\end{figure}
   \item  Let $ w_6(\lambda,x) = (x\lambda)/\log(1+x\lambda) $.
  It is easy to check that $x\mapsto w_6(\lambda, x)$ is non-decreasing
  for every $\lambda > 0$.
\begin{figure}[h!]
\includegraphics[width=10cm]{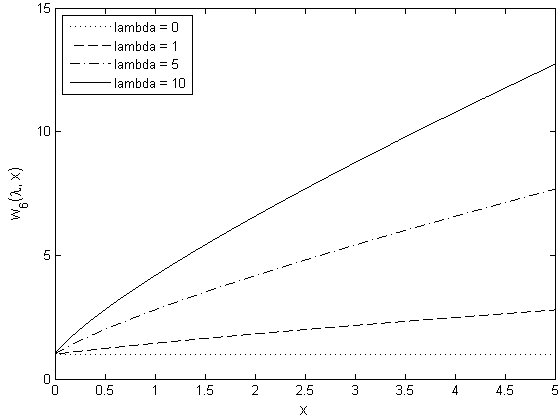}
\caption{The function $x\mapsto w_6(\lambda, x)$. When $x\downarrow 0$,
then $w_6(\lambda, x)$ converges to $1$ for every $\lambda >0$.}
\label{fig-w6}
\end{figure}
    \item
    \label{example1.1-7}
    Let $ w_7(\lambda,x) =  \frac{\log(1+x+\lambda)}{x+\lambda}\frac{x}{\log(1+x)}$.
   The fact that $x\mapsto w_7(\lambda, x)$ is non-decreasing
   for every $\lambda > 0$ is proved in Lemma~\ref{lma-64}.
\begin{figure}[h!]
\includegraphics[width=10cm]{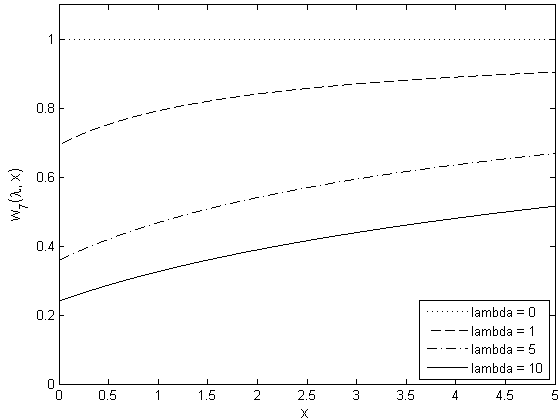}
\caption{The function $x\mapsto w_7(\lambda, x)$. When $x\downarrow 0$,
then $w_7(\lambda, x)$ converges to $\log(1+\lambda)/\lambda $
(same limit as in Figure \ref{fig-w5}) for every $\lambda >0$.}
\label{fig-w7}
\end{figure}
 \end{enumerate}
\end{example}

From mathematical definitions and accompanying graphs,
we see monotonic features of the functions $w_i(\lambda, x)$ with
respect to $\lambda $. Specifically, for every pair $\theta < \lambda $,
we check -- in most cases quite easily -- the following properties:
\begin{itemize}
  \item $w_1(\theta, x) < w_1(\lambda, x)$
  for all $x \in (0,\infty)$

  \item $w_2(\theta, x) \ge w_2(\lambda, x)$
  for all $x \in (0,\infty)$

  \item $w_3(\theta, x) > w_3(\lambda, x)$
  for all $x \in (0,\infty)$

  \item $w_4(\theta, x) < w_4(\lambda, x)$
  for all $x \in (0,\infty)$, which follows from the easy-to-verify inequality
  $((1+x)^\theta -1 )/\theta < ((1+x)^\lambda -1)/\lambda$

  \item
  $w_5(\theta, x) > w_5(\lambda, x)$ for all  $x \in (0,1)$
  \item[]
  $w_5(\theta, x) = w_5(\lambda, x)=1$ when  $x =1$
  \item[]
  $w_5(\theta, x) < w_5(\lambda, x)$ for all  $x > 1$

  \item $w_6(\theta, x) < w_6(\lambda, x)$
  for all $x \in (0,\infty)$

  \item $w_7(\theta, x) > w_7(\lambda, x)$
  for all $x \in (0,\infty)$; see Lemma~\ref{lma-64}
\end{itemize}
These monotonicity properties play important roles when
establishing strict monotonicity of the corresponding functions
$\lambda \mapsto H[\lambda,X]$ in Section \ref{section-4}.
They are also helpful and indeed decisive
when determining the set $\Lambda[w,X]$ of those $\lambda>0$
for which the premium $H[\lambda,X]$ is well-defined and finite.

Certainly, in the context of the present section, it is important
to check log-supermodularity of the seven functions
of Example \ref{example1.1}. This makes the content of the following theorem.

\begin{theorem} \label{th-logsuper}
The seven function $(\lambda, x)\mapsto w_i(\lambda, x)$, $1\le i \le 7$,
are log-supermodular.
\end{theorem}

\begin{proof}
In the case of the first two functions of
Example \ref{example1.1}, we shall use the noted (and proved) fact that
$(\lambda, x)\mapsto w(\lambda, x)$ is log-supermodular if and only if
for every pair $\theta \le \lambda$ there is a non-decreasing function
$h(x) = h_{\theta, \lambda}(x)$ such that equation (\ref{condition-1}) holds.
In the case of the remaining five functions
of Example \ref{example1.1}, we shall use the fact that the function
$(\lambda, x)\mapsto w(\lambda, x)$ is log-supermodular if and only if
$(\partial^2/ \partial \lambda \partial x) L(\lambda, x) \le 0$
for all $(\lambda, x) $. Recall that $L(\lambda, x) = - \log( w(\lambda, x))$.
\begin{enumerate}
  \item
  The function $w_1(\lambda, x)$ is log-supermodular.
  This follows from decomposition (\ref{condition-1}) with
  $h(x) = e^{(\lambda-\theta)x}$, which is an increasing function
  whenever $\theta < \lambda$.

  \item
  The function $w_2(\lambda, x)$ is log-supermodular,
  which follows from decomposition (\ref{condition-1}) with
$h(x) = \mathbf{1}{\{x> \lambda \}}$, which is non-decreasing.

    \item
    The log-supermodularly of $w_3(\lambda, x)$ follows from
    $(\partial^2/ \partial \lambda \partial x) L(\lambda, x) \le 0$.
    Indeed, since
    $L(\lambda, x) = -\log(1-e^{-x/\lambda})$, we have that
    $$
    \frac{\partial^2}{\partial \lambda \partial x} L(\lambda, x)
    = -\frac{e^{-x/\lambda}(e^{-x/\lambda}-1+x/\lambda)}{\lambda^2(e^{-x/\lambda}-1)^2} < 0.
    $$

 \item
 The function $w_4(\lambda, x)$ is log-supermodular, for which we first check the equation
\begin{align}
- & (x+1) \lambda^2 \Big(e^{-((1+x)^\lambda - 1)/\lambda} \Big)
\Big( e^{((x+1)^{\lambda}-1)/\lambda}-x\Big)^2
 \frac{\partial^2}{\partial \lambda \partial x}L(\lambda, x)
\notag
\\
&= \lambda^2 (x+1)^\lambda \log(x+1)
\Big(e^{((x+1)^\lambda-1)/\lambda}-x\Big)
\notag
\\
&\quad +\big(\lambda(x+1)^\lambda \log(x+1)-(x+1)^\lambda+1\big)
\big(1-((x+1)^\lambda-1) x \big).
\label{lm-61}
\end{align}
The right-hand side of equation (\ref{lm-61}) is positive
for all $\lambda > 0$ and $x > 0$. This we formulate as
Lemma \ref{lma-61} and prove in Section \ref{section-6}.

  \item
  The function $w_5(\lambda, x)$ is log-supermodular because
    $(\partial^2/ \partial \lambda \partial x) L(\lambda, x) \le 0$,
    which we prove as follows:
  $$
   \frac{\partial^2}{\partial \lambda \partial x} L(\lambda, x)
   = \frac{(1+\lambda )^{x-1}}{((1+\lambda)^x-1)^2}
   \Big(x\log(1+\lambda) + 1 - (1+\lambda)^x \Big).
  $$
 Since $x\log(1+\lambda) + 1 - (1+\lambda)^x < 0$ for all $\lambda > 0$ and $x > 0$,
 this establishes the result.
  \item
  The function $w_6(\lambda, x)$ is log-supermodular
  because
    $(\partial^2/ \partial \lambda \partial x) L(\lambda, x) \le 0$. Indeed, since
  $L(\lambda, x) = -\log(x\lambda / \log(1+x\lambda))$
  we have that
  $$
   \frac{\partial^2}{\partial \lambda \partial x} L(\lambda, x)
   = \frac{\log(1+x \lambda)-x\lambda}{(\log(1+x\lambda)
   (1+x \lambda))^2} < 0.
  $$

  \item
  The function $w_7(\lambda, x)$ is log-supermodular because
    $(\partial^2/ \partial \lambda \partial x) L(\lambda, x) \le 0$,
    which we prove by first establishing the equation:
  \begin{align}
   - &
   \Big ((1+x+\lambda)(x+\lambda)\log(1+x+\lambda)\Big )^2
   \frac{\partial^2}{\partial \lambda \partial x} L(\lambda, x)
   \notag
   \\
   &= \big ((1+x+\lambda) \log(1+x+\lambda) \big )^2
   -(x+\lambda)^2 \log(1+x+\lambda)
   -(x+\lambda)^2 .
  \label{lm-62}
  \end{align}
The right-hand side of equation (\ref{lm-62}) is positive,
which we formulate as
Lemma \ref{lma-62} and prove in Section \ref{section-6}.
\end{enumerate}
This completes the proof of Theorem \ref{th-logsuper}.
\end{proof}

We now reflect upon the above proof
in the context of the easily checked fact that:
\begin{itemize}
  \item $L(\lambda, x) = \alpha f(\lambda^\alpha x)$ is submodular
for any non-increasing and concave function $f$ and any real
number $\alpha $.
\end{itemize}
Hence, submodularity of $-\log(w_1(\lambda, x))$
follows by choosing  $\alpha =1$ and $f(t)=-t$.
Note, however, that submodularity of $-\log(w_3(\lambda, x))$
does not follow from such arguments since, with $\alpha =-1$,
the function $f(t)=\log(1-e^{-t})$ is increasing, though concave.
Submodularity of $-\log(w_6(\lambda, x))$
does not follow from the arguments either because, with $\alpha =1$,
the function $f(t)=-\log(t/\log(1+t))$ is convex, though decreasing.
Finally we note that submodularity of the weight functions
$-\log(w_i(\lambda, x))$, $i\in \{ 5,6,7\}$, does not follow
from any of the constructions given in Table I on p.~312 of Topkis (1978).

\section{Loading strict-monotonicity}
\label{section-4}

To prove that the function $\lambda \mapsto
H[\lambda,X]$ is (strictly) increasing, we need additional
assumptions.

\begin{assumption}\rm \label{assump-1a}
For every $x$, the function $\lambda \to w(\lambda, x)$ is monotonic.
\end{assumption}

The assumption does not require $\lambda \to w(\lambda, x)$
to be increasing, nor even non-decreasing, which at first sight
might be somewhat surprising. However, recall the function
$\lambda \mapsto w_5(\theta, x)$. It might be increasing, decreasing,
or constant, depending on the value of $x>0$.

Our next assumption imposes a kind of strict monotonicity on the function
$\lambda \to w(\lambda, x)$ by requiring, loosely speaking, that if
there is a point $x$ such that the function is constant, then the point
should not be `visible', that is, the loss variable $X$ should not
take on the value $x$ almost surely.

\begin{assumption}\rm \label{assump-1b}
For every pair $\theta \ne \lambda $, the set
$\big \{ x \in [0,\infty ) :\, w(\theta, x)=w(\lambda, x) \big \}$
has  $F$-measure zero.
\end{assumption}

To work out some intuition concerning Assumption \ref{assump-1b},
we look at the seven functions $w_i(\lambda, x)$ of Example \ref{example1.1}.
\begin{itemize}
  \item
For every $x>0$,  $w_1(\theta, x) \ne
w_1(\lambda, x)$ whenever $\theta \ne \lambda$.
Same holds for $w_3(\lambda, x)$, $w_4(\lambda, x)$,
and $w_6(\lambda, x)$.
  \item
We have $w_2(\theta, x) = w_2(\lambda, x)$ for every $x\not\in (\theta ,
\lambda]$ when $\theta < \lambda$.
This implies that Assumption \ref{assump-1b} is {\it not} satisfied.
\item
For every $x \in (0,\infty )\setminus \{1\}$, we have
$w_5(\theta, x) \ne w_5(\lambda,x)$ whenever $\theta \ne \lambda$.
When $x =1$, then $w_5(\theta, x) = w_5(\lambda,x)$, and we thus need to
assume $\mathbf{P}[X=1]=0$ in order to make the point $x=1$ `invisible'.
This is equivalent to assuming
the continuity of $F$ at the point $x =1$.
\item
For every $x > 0$, we have $w_7(\theta, x)
\ne w_7(\lambda,x)$ whenever $\theta \ne \lambda$.
\end{itemize}

The following assumption requires, roughly speaking, the existence of
a point in the closure of the half-line $(0,\infty )$ such that
all the functions $x\mapsto w(\lambda, x)$ take on one and same
positive value at the point.

\begin{assumption}\rm \label{assump-1c}
For every pair $(\theta , \lambda ) $ there exists a point $x_0\in [0,\infty ]$
such that $w(\lambda, x_0) =w(\theta, x_0)>0$ and
the functions $x\mapsto w(\lambda, x)$ and $x\mapsto
w(\theta, x)$ are either both left-continuous or both
right-continuous at $x_0$. (If $x_0=0$, then they have to be
right-continuous, while if $x_0=\infty $ they have to be
left-continuous.)
\end{assumption}

Even though the above assumption may look somewhat artificial,
and perhaps even strange, it is nevertheless satisfied by
the seven functions $w_i(\lambda, x)$ of Example \ref{example1.1}:
\begin{itemize}
  \item
  For every $i\in \{1,4,6\}$, we have
$w_i(\lambda, 0)=1$ for every $\lambda > 0$. Thus, $x_0=0$.
These functions $x\mapsto w_i(\lambda,x)$
are continuous everywhere on $(0,\infty )$.

  \item
For every $i \in \{2,3,7\}$, we have $w_i(\lambda, \infty)=1$ for
every $\lambda > 0$. Thus, $x_0=\infty $.
The function $x\mapsto w_2(\lambda, x)$
 is left-continuous on $(0,\infty )$,
 and the functions $x\mapsto w_i(\lambda, x)$ for $i\in \{3,7\}$
are continuous everywhere on $(0,\infty)$.

  \item
  We have $w_5(\lambda, 1)=1$ for every $\lambda > 0$. Thus, $x_0=1$.
  The function $x\mapsto w_5(\lambda, x)$ is continuous everywhere
  on $(0,\infty )$.

\end{itemize}

\begin{theorem} \label{th-1a}
When $(\lambda, x)\mapsto w(\lambda, x)$ is log-supermodular
and Assumptions \ref{assump-1a}--\ref{assump-1c} are satisfied, then
$\lambda \mapsto H[\lambda,X]$ is increasing.
\end{theorem}

\begin{proof}
We prove by contradiction. Let there be $\theta < \lambda $ such
that $H[\theta,X]=H[\lambda,X]$. This can be rewritten as
$\mathbf{E}[X_{w(\theta,\cdot)}]=\mathbf{E}[X_{w(\lambda, \cdot)}]$.
On the other hand, from bounds (\ref{w-cdf-2}) we have that
$F_{w(\lambda,\cdot)}\le F_{w(\theta, \cdot)}$ for all $x\ge 0$.
Hence, the cdf's $F_{w(\lambda, \cdot)}$ and $F_{w(\theta,\cdot)}$
must coincide, except possibly on a set of
Lebesque measure zero. Since these functions are
right-continuous, they should therefore coincide for all $x\ge 0$.
In other words, we have the equation
\begin{equation}
\label{distr-eq}
\frac{ \mathbf{E} [ \mathbf{1}\{X\le x\} w(\lambda, X) ]
}{\mathbf{E}[w(\lambda,X)]} = \frac{\mathbf{E} [
\mathbf{1}\{X\le x\} w(\theta, X)]}{\mathbf{E}[w(\theta, X)]}
\end{equation}
for all $x\ge 0$, and thus, in turn,
\begin{equation} \label{distr-eq1}
{w(\lambda, x) \over \mathbf{E}[w(\lambda, X)]}
={w(\theta, x) \over \mathbf{E}[w(\theta, X)]}
\end{equation}
for $F$-almost all $x\ge 0$.
Combining the latter equation with the right- or left-continuity postulated
in Assumption \ref{assump-1c}, we have
that equation (\ref{distr-eq1}) must hold at $x=x_0$. Since
$w(\lambda, x_0)=w(\theta, x_0)$ by Assumption
\ref{assump-1c}, the expectations $\mathbf{E}[w(\lambda, X)]$
and $\mathbf{E}[w(\theta, X)]$ coincide. But then
equation (\ref{distr-eq1}) says that $w(\lambda, x)=w(\theta, x)$
for $F$-almost all $x \ge 0$. This contradicts Assumption \ref{assump-1b}
and thus finishes the proof of Theorem \ref{th-1a}.
\end{proof}

\begin{corollary} \label{corlll-1a}
For every weight function $(\lambda, x)\mapsto w_i(\lambda, x)$,
$i\in \{1,\dots , 7\} \setminus \{2\}$, the corresponding function
$\lambda \mapsto H[\lambda,X]$ is increasing.
\end{corollary}

The above corollary excludes the weight
function $w_2(\lambda, x)=\mathbf{1}\{x> \lambda \}$ because
it does not satisfy Assumption \ref{assump-1b}.
Nevertheless, we have the following result.

\begin{proposition}
\label{prop-1}
Let $w_2(\lambda, x)=\mathbf{1}\{x> \lambda \}$. If the
cdf $F$ is increasing on $(0,\infty)$, then
the function $\lambda \mapsto H[\lambda,X]$ is increasing.
\end{proposition}

\begin{proof}
Equation (\ref{distr-eq}) with the weight function
$w_2(\lambda, x)=\mathbf{1}\{x> \lambda \}$ becomes
\begin{align}
\label{distr-eq-rewritten}
\frac{\mathbf{P}[\lambda < X \le x]}{1-F(\lambda)}
= \frac{\mathbf{P}[\theta < X \le x]}{1-F(\theta)},
\end{align}
where $\theta < \lambda$. Since $F$ is increasing on $(0,\infty)$,
the denominators on both sides of equation (\ref{distr-eq-rewritten})
are non-zero, that is, positive. But the numerator of the left-hand side
is equal to $0$ for all $x\in (\theta, \lambda ]$. Hence, the right-hand
side should also be equal to $0$, which means that
$F(x)- F(\theta)=0$ for all $x\in (\theta, \lambda ]$, and
thus $F(\lambda )= F(\theta)$ in particular.
Since $\theta < \lambda$, the latter equation contradicts
the assumption that  $F$ is increasing, thus concluding the proof
of Proposition \ref{prop-1}.
\end{proof}

\section{Continuity-type results}
\label{section-5}

The motivating actuarial problem at the end of Section \ref{section-1} demonstrates
the importance of establishing continuity-type results for the function
$\lambda \mapsto H[\lambda,X]$. These are naturally connected with
continuity-type properties of the function
$\lambda \mapsto w(\lambda, x)$. In this context, we first look
at the seven functions $w_i(\lambda, x)$ of Example \ref{example1.1}:
\begin{itemize}
  \item
The functions $\lambda \mapsto w_i(\lambda, x)$,
$i\in \{1,\dots , 7\} \setminus \{2\}$,
are continuous for every $x> 0$.
  \item
The function $\lambda \mapsto w_2(\lambda, x)$ is
right-continuous for every $x> 0$.
\end{itemize}

\begin{assumption}\rm \label{assump-2}
Let $\lambda_0 >0 $ be such that
$\displaystyle\lim_{\lambda \downarrow \lambda_0} w(\lambda, x)=
w(\lambda_0, x) $ for $F$-almost all $x > 0$.
\end{assumption}

All the seven functions of Example \ref{example1.1} satisfy
Assumption \ref{assump-2} for every $\lambda_0 >0$. To connect
this assumption with the right-continuity of the function
$\lambda \mapsto H[\lambda,X]$, we need to interchange
limit and integration operations. This is taken care by the following
assumption.

\begin{assumption}\rm \label{assump-4}
Let there exist a random variable $Y\ge 0$ such that
$\mathbf{E}[Y]<\infty $ and $X w_{\lambda }(X)\le Y$ for all
$\lambda $ in a (small) neighbourhood of $\lambda_0>0$.
\end{assumption}

The seven functions of Example \ref{example1.1} satisfy Assumption \ref{assump-4}
under the following conditions:
\begin{itemize}
  \item
For $w_i(\lambda, x)$, $i\in\{1,4,5,6\}$, the moment
$\mathbf{E}[X w_i(\lambda_0+\epsilon, X)]$ is finite for some
$\epsilon >0$.

\item
For $w_i(\lambda, x)$, $i \in \{2,3,7\}$, the moment $\mathbf{E}[X]$ is finite.
\end{itemize}

The next theorem now becomes obvious and its proof is omitted.

\begin{theorem} \label{th-2a}
If Assumptions \ref{assump-2} and \ref{assump-4} are satisfied,
then the function $\lambda \mapsto H[\lambda,X]$ is
right-continuous at $\lambda_0$.
\end{theorem}

We next investigate the continuity of the function $\lambda \mapsto H[\lambda,X]$.

\begin{assumption}\rm \label{assump-3}
Let $\lambda_0>0$ be such that $
\displaystyle\lim_{\lambda \to \lambda_0} w(\lambda, x)=
w(\lambda_0, x)$ for $F$-almost all $x > 0$.
\end{assumption}

In the case of the seven functions $w_i(\lambda, x)$
of Example \ref{example1.1}, we have the following notes:
\begin{itemize}
  \item
The functions $\lambda \mapsto w_i(\lambda, x)$,
$i\in \{1,\dots , 7\} \setminus \{2\}$,
are continuous for
every $x> 0$ and so Assumption \ref{assump-3} is satisfied for
all $\lambda_0>0$.
  \item
The function $\lambda \mapsto w_2(\lambda, x)$ is
right-continuous. Assumption \ref{assump-3}
fails because the function $\lambda \mapsto w_2(\lambda, x)$
has a jump (of size $1$) at the point $x$.
However, if $F$ is continuous on $(0,\infty )$, then
every singleton $\{ x \}$ has $F$-measure zero and thus
Assumption \ref{assump-3} is satisfied for every $\lambda_0>0$.
 \end{itemize}

The proof of the next theorem is elementary and thus omitted.

\begin{theorem} \label{th-2b}
If Assumptions \ref{assump-4} and \ref{assump-3} are satisfied, then
the function $\lambda \mapsto H[\lambda,X]$ is continuous at
$\lambda_0$.
\end{theorem}

We next present an example showing that the continuity of $F$ in the case
of the weight function $w_2(\lambda, x)$ is crucial for the function
$\lambda \mapsto H[\lambda,X]$ to be continuous.
Namely, let $F$ be the empirical cdf $F_n$
based on a sample $x_1, \dots, x_n$. Denote
the corresponding order statistics by $x_{1:n}\le \cdots \le x_{n:n}$.
The function $ \lambda \mapsto H[\lambda,X]$ is not continuous.
Indeed, it is only right-continuous: for every $1\le i \le n$, the function
$ \lambda \mapsto H[\lambda,X]$ takes on the value
$(x_{i:n}+\cdots +x_{n:n})/(n-i+1)$ when
$\lambda \in [x_{(i-1):n}, x_{i:n})$, with the notation
$x_{0:n}=0$.

\section{Technicalities}
\label{section-6}

Here we have collected technical details that were left out
from previous sections.

\begin{lemma}
\label{aux-lem-1}
We have $(y+1) \log (y+1) > y > \log(y+1)$ for every $y > 0$.
\end{lemma}

We shall use Lemma \ref{aux-lem-1} a number of times in this section.
We omit its proof as it is elementary.

\begin{lemma}
\label{lma-61}
The right-hand side of equation (\ref{lm-61}) is positive
for all $\lambda > 0$ and $x > 0$.
\end{lemma}

\begin{proof}
With the notation $y=((x+1)^\lambda-1)/\lambda $,
the right-hand side of equation (\ref{lm-61}) is positive if and only if
\begin{multline}
\lambda (\lambda y+1)  \log(\lambda y +1)\Big(e^y -(\lambda y + 1)^{1/\lambda} + 1 \Big)
\\
+\big( (\lambda y+1) \log (\lambda y+1) - \lambda y\big)\big((\lambda y+1)
- \lambda y (\lambda y + 1)^{1/\lambda} \big) > 0.
\label{bound-61}
\end{multline}
Note that since $\lambda >0$, we have $y>0$ if and only if $x>0$. Hence, we need
to show that bound (\ref{bound-61}) holds for all $\lambda > 0$ and $y>0$.
Since $(\lambda y+1) \log (\lambda y+1) > \lambda y $
(see Lemma~\ref{aux-lem-1}), we have that bound (\ref{bound-61})
follows from
\[
(\lambda y+1)  \log(\lambda y +1)\Big(e^y -(\lambda y + 1)^{1/\lambda} \Big)
-\big( (\lambda y+1) \log (\lambda y+1) - \lambda y\big)y (\lambda y + 1)^{1/\lambda} > 0,
\]
which can be rewritten as
\begin{equation}
\label{rel-error}
e^y > (\lambda y + 1)^{1/\lambda}
\bigg(1+y -\frac{\lambda y^2}{ (\lambda y+1) \log (\lambda y+1)} \bigg).
\end{equation}
This is an interesting bound on its own, and we shall dwell upon it once
the proof of Lemma \ref{lma-61} has been finished. To prove bound (\ref{rel-error}),
we apply the logarithmic function and see that the bound is equivalent
to $g(y)>0$, where
\[
g(y)=y - \frac{\log(\lambda y + 1)}{\lambda}-\log \bigg(1+y -\frac{\lambda y^2}{ (\lambda y+1)
\log (\lambda y+1)} \bigg) .
\]
The function $g(y)$ converges to $0$ when $y\downarrow 0$, and so to prove
its positivity, we show that it is strictly increasing.
For this, we first verify the equation
\[
(\lambda y+1)\log(\lambda y+1) g'(y)
= \frac{\big((\lambda y+1)\log(\lambda y+1)
- \lambda y\big)\big(\log(\lambda y+1) \lambda y^2
+ \lambda y -\log(\lambda y+1)\big)}{(\lambda y+1)
\log(\lambda y+1)+y\big((\lambda y+1) \log(\lambda y+1)-\lambda y\big)}.
\]
Using Lemma~\ref{aux-lem-1}, we check that
the numerator and denominator in the above ratio are positive.
This implies that $g'(y)>0$ and concludes the proof of Lemma \ref{lma-61}.
\end{proof}

The above proof of log-supermodularity of $w_4(\lambda, x)$
contains an interesting element, which is related to bound (\ref{rel-error}).
Namely, it is well-known from Calculus that, for every $y > 0$, the function
$\lambda \mapsto  (\lambda y + 1)^{1/\lambda}$
is decreasing and
$\lim_{\lambda \downarrow 0} (\lambda y + 1)^{1/\lambda} = e^y$.
With the establishment of bound (\ref{rel-error}),
we have obtained a lower bound for the relative error
of this approximation of $e^y$. Namely, we have that
\begin{equation}
\frac{e^y - (\lambda y + 1)^{1/\lambda}}{ (\lambda y + 1)^{1/\lambda}}
> y\bigg (1 - \frac{\lambda y}{(\lambda y+1)\log(\lambda y+1)} \bigg ) > 0
\label{interesting}
\end{equation}
for all $\lambda > 0$ and $y > 0$.

\begin{lemma}
\label{lma-62}
The right-hand side of equation (\ref{lm-62}) is positive for
all $\lambda > 0$ and $x > 0$.
\end{lemma}

\begin{proof}
Substituting $y=x+\lambda$, we obtain that the
expression on the right-hand side of equation (\ref{lm-62}) is positive
if and only if $f(y)>0$, where
$$
f(y)=\big ((1+y)\log(1+y)\big )^2 - y^2\log(1+y) -y^2 .
$$
Since $f(y)=0$ when $y=0$, we have $f(y)>0$ for all $y> 0$
if the function $f$ is strictly increasing. For this, we check
the equation $(1+y)f'(y)=g(y)$, where
$$
g(y)= 2\big ((1+y)\log(1+y)\big )^2 -3y^2+2(1+y)\log(1+y) -2y.
$$
To show that $g(y)>0$ for all $y>0$, we note that
$g(0)=0$ and then show that $g$ is strictly increasing.
For this, we check the equation
$$
(1+y)g'(y) =
4\big ((1+y)\log(1+y)\big )^2
+4y(1+y) \log(1+y)-6y^2+6(1+y)\log(1+y)-6y.
$$
The right-hand side of the equation
is positive for all $y > 0$, which follows from the bound
$(y+1) \log (y+1) > y $ (see Lemma~\ref{aux-lem-1}).
This concludes the proof of Lemma \ref{lma-62}.
\end{proof}

\begin{lemma}
\label{lma-64}
The function $(\lambda, x)\mapsto w_7(\lambda, x)$
is increasing in $x$ and decreasing in $\lambda$.
\end{lemma}

\begin{proof}
First, a simple calculation shows that
\begin{align*}
&
(1+x)(1+x+\lambda)(x+\lambda)^2\big (\log(1+x)\big )^2
\frac{\partial}{\partial x} w_7(\lambda, x)
\\
&\hspace{1cm}= x(1+x)(x+\lambda)\log(1+x)+\lambda (1+x)(1+x+\lambda)
\log(1+x)\log(1+x+\lambda) \\
&\hspace{1.5cm} -x(x+\lambda)(1+x+\lambda)\log(1+x+\lambda).
\end{align*}
From this equation we conclude that
\begin{align*}
\lim_{\lambda \rightarrow 0}\frac{\partial }{\partial x} w_7(\lambda, x) &=0,
\\
\lim_{x \rightarrow 0} \frac{\partial }{\partial x} w_7(\lambda, x)
&= \frac{2\lambda-\lambda \log(1+\lambda)-2\log(1+\lambda)
+\lambda^2\log(1+\lambda)}{2\lambda^2(1+\lambda)} > 0.
\end{align*}
We have shown in Lemma \ref{lma-62} that
$(\partial^2/ \partial \lambda \partial x)  \log(w_7(\lambda, x)) > 0$,
which can be rewritten as
$$
\frac{\partial}{\partial \lambda} \bigg(\frac{(\partial /\partial x)
w_7(\lambda, x)}{w_7(\lambda, x)}  \bigg) > 0,
$$
thus implying that the function
$\lambda \mapsto ((\partial /\partial x)w_7(\lambda, x))/w_7(\lambda, x)$
is strictly increasing. But the function takes on the value $0$ at $\lambda =0$,
and thus the function is positive for all
$x > 0$ and $ \lambda > 0$.
Since the denominator $w_7(\lambda, x)$ is strictly positive,
we have $(\partial /\partial x) w_7(\lambda, x) > 0$
for all $x > 0$ and $ \lambda > 0$.
Consequently, $w_7(\lambda, x)$ is increasing in $x$.

To show that $w_7(\lambda, x)$ is decreasing in $\lambda$, we check the equation
\begin{equation}
(1+x+\lambda)(x+\lambda)^2\log(1+x)
\frac{\partial}{\partial \lambda} w_7(\lambda, x)
= x(x+\lambda-(1+x+\lambda)\log(1+x+\lambda)) .
\label{final}
\end{equation}
Lemma~\ref{aux-lem-1} implies that the right-hand side of equation (\ref{final})
is negative. This concludes the proof of Lemma \ref{lma-64}.
\end{proof}

\section*{Acknowledgments}

The research was partially supported by
the Natural Sciences and Engineering Research Council (NSERC) of Canada.

\section*{References}
\def\hang{\hangindent=\parindent\noindent}

\hang
Broll, U., Egozcue, M., Wong, W.K.\, and Zitikis, R.(2010).
Prospect theory, indifference curves, and hedging risks.
Applied Mathematics Research Express  (to appear).

\hang  Denuit, M., Dhaene, J., Goovaerts, M.J., Kaas, R., 2005.
Actuarial Theory for Dependent Risk: Measures, Orders and Models.
Wiley, New York.

\hang Denneberg, D. (1994).
Non-Additive Measure and Integral.
Kluwer, Dordrecht.

\hang
Egozcue, M., Fuentes Garc\'ia, L., Wong, W.K.\, and Zitikis, R. (2010).
The covariance sign of transformed
random variables with applications to economics and finance.
IMA Journal of Management Mathematics (to appear).

\hang
Fujishige, S. (1991). Submodular Functions and Optimization.
North-Holland, Amsterdam.

\hang Furman, E.\, and Zitikis, R. (2008). Weighted premium
calculation principles. Insurance: Mathematics and
Economics, 42, 459--465.

\hang Furman, E.\, and Zitikis, R. (2009).
Weighted pricing functionals with applications to insurance: an overview.
North American Actuarial Journal, 13, 483--496.

\hang Jones, B.L.\, and Zitikis, R. (2007). Risk measures,
distortion parameters, and their empirical estimation.
Insurance: Mathematics and Economics, 41, 279--297.

\hang
Kamps, U. (1998).
On a class of premium principles including the Esscher principle.
Scandinavian Actuarial Journal, 1998, 75--80.

\hang Lehmann, E.L. (1966). Some concepts of dependence.
Annals of Mathematical Statistics, 37, 1137--1153.

\hang
Mitrinovi\'{c}, D.S., Pe\v{c}ari\'{c}, J.E. and Fink, A.M. (1993).
Classical and New Inequalities in Analysis.
Kluwer, Dordrecht.

\hang
Narayanan, H. (1997). Submodular Functions and Electrical Networks.
North-Holland, Amsterdam.

\hang
Pe\v{c}ari\'{c}, J.E., Proschan, F. and Tong, Y.L. (1992).
Convex Functions, Partial Orderings, and Statistical Applications.
Academic Press, Boston.

\hang
Topkis, D.M. (1978). Minimizing a submodular function on a lattice.
Operations Research, 26, 305--321.

\hang
Topkis, D.M. (2001). Supermodularity and Complementarity.
Princeton University Press, Princeton.

\end{document}